\documentclass[conference]{IEEEtran}
\usepackage{cite}
\usepackage[cmex10]{amsmath}
\usepackage{amsfonts,amssymb}
\usepackage{amsmath,amssymb}
\usepackage{amsthm}
\usepackage{color}
\usepackage[mathscr]{eucal}
\usepackage{graphics,graphicx,multicol}
\usepackage{epsfig}
\usepackage{enumerate}
\usepackage{subfigure}
\usepackage{algorithmic}
\usepackage[section]{algorithm}
\usepackage{morefloats}
\usepackage{enumerate}
\usepackage{subfigure}
\usepackage{multirow}
\usepackage{xfrac}
\usepackage{array}
\usepackage{pstricks, pst-node, pst-plot, pst-circ}
\usepackage{pst-3dplot,pstricks-add,pst-solides3d,pst-grad}
\usepackage{moredefs}
\usepackage{cite}
\usepackage{courier}
\newtheorem{thm}{Theorem}[section] 
\newtheorem{lem}[thm]{Lemma}

\newtheorem{rem}[thm]{Remark}
\providecommand{\mbf}[1]{\mathbf{#1}}						
\providecommand{\wt}[1]{\widetilde{#1}}					

\providecommand{\bsym}[1]{\boldsymbol{#1}}				
\providecommand{\bsymwt}[1]{\widetilde{\boldsymbol{#1}}}	
\providecommand{\mbb}[1]{\mathbb{#1}}

\begin{document}
\title{Network MIMO with Partial Cooperation between Radar and Cellular Systems}
\author{Ahmed Abdelhadi and T. Charles Clancy \\
Hume Center, Virginia Tech, Arlington, VA, 22203, USA\\
\{aabdelhadi, tcc\}@vt.edu
}
\maketitle

\begin{abstract}
To meet the growing spectrum demands, future cellular systems are expected to share the spectrum of other services such as radar.  In this paper, we consider a network multiple-input multiple-output (MIMO) with partial cooperation model where radar stations cooperate with cellular base stations (BS)s to deliver messages to intended mobile users. So the radar stations act as BSs in the cellular system. However, due to the high power transmitted by radar stations for detection of far targets, the cellular receivers could burnout when receiving these high radar powers. Therefore, we propose a new projection method called small singular values space projection (SSVSP) to mitigate these harmful high power and enable radar stations to collaborate with cellular base stations. In addition, we formulate the problem into a MIMO interference channel with general constraints (MIMO-IFC-GC). Finally, we provide a solution to minimize the weighted sum mean square error minimization problem (WSMMSE) with enforcing power 
constraints on both radar and cellular stations.  
\end{abstract}
\IEEEpeerreviewmaketitle


\section{Introduction}

 Federal Communications Commission (FCC) and the National Telecommunications and Information Administration (NTIA) studies show very low utilization of huge chunks of spectrum held by the federal agencies, especially in urban areas. Meanwhile, there is a very heavy utilization of spectrum held by commercial operators, e.g. cellular operators, in these urban areas. President's Council of Advisors on Science and Technology (PCAST) recommendations in order to efficiently utilize federal spectrum are to share federal spectrum with commercial operators \cite{PCAST12}. The sharing will result in enormous economic and social advances for the nation. Meanwhile, this sharing should not endanger the main mission of federal incumbents, e.g. sharing radar spectrum should not affect its target tracking capabilities. Therefore, new approaches should be developed with these considerations in mind.

A recent report by NTIA \cite{NTIA12} concluded that sharing radar spectrum with WiMAX requires huge exclusion zones up to tens of kilometers to protect the WiMAX receivers from harmful interference signal transmitted by radar. This is due to WiMAX receivers are designed to handle low power levels in the range of Watts while the power transmitted by radar is in the range of Kilo and Mega Watts. This includes shipborne radars that are deployed on military ships on the east and west coasts of the United States. Which in turn results in depriving these areas, i.e. where the majority of the US population live, from the benefits of sharing radar spectrum. 

On the other hand, within the cellular system, interference is a major obstacle against achieving the spectral efficiency expected from developed multiple-antenna techniques \cite{linear_precoding_journal}. It is shown in \cite{DMP_Inter1,MIMO_Inter} that multiple-input multiple-output (MIMO) capacity gains are deteriorated due to inter-cellular interference. In a radar/cellular coexistence scenario, radar receivers have highly sensitive receivers for detecting reflected signals from far targets. Therefore, it is highly susceptible to interference from commercial wireless system operating on radar bands. In the past, radar has been guaranteed exclusive rights to radio spectrum allocation to avoid its operation from being affected by commercial wireless systems interference \cite{NTIA12, KAC14WTS}. Therefore, a radar/cellular network-level interference management is of fundamental importance to sustain the radar/cellular coexistence along with limiting inter-cellular interference and harnessing the advantages 
of cellular MIMO technology. 

To address the aforementioned challenges, we propose a novel coexistence scenario and model between radar and cellular system. In this model, the radar signal is steered to null-space plus small singular values space of the interference channel between the radar and cellular system. The approach benefits both radar and cellular systems. On the radar side, it will increase the projection space dimensions and therefore radar performance metrics are improved compared to projection with smaller dimensions, e.g. null space projection \cite{KAC+14ICNC}, see \cite{KAC_QPSK} for more details. On the cellular side, this approach suppresses the high power of radar in the direction of cellular network so it does not burn out the cellular receivers. In addition, the transmitted radar signal could be used to transmit communication messages to enhance the overall system performance and quality of service (QoS) of cellular system. In our model, we propose network MIMO with partial cooperation for merging radar stations in 
the cellular network.

\subsection{Related Work}

To benefit from radar spectrum, researchers have proposed the use of spatial domain to mitigate MIMO radar interference to communication system \cite{KAC+14ICNC}. One of the studies proposed projection of radar signal into the null space of the interference channel between radar and communication systems \cite{KAC14_MILCOM}. In another study, researchers designed radar waveform that doesn't cause harmful interference with communication system and successfully achieves the radar mission objectives \cite{KAC14DySPANWaveform}. In the past, sharing of government bands has been allowed for commercial wireless systems under the condition of low power transmission in order to protect incumbent from harmful interference \cite{KAC_IEEE_Sensors}. Famous examples are WiFi and Bluetooth at 2450-2490 MHz band, wireless local area network (WLAN) at the 5.25-5.35 and 5.47-5.725 GHz \cite{FCC_5GHz_Radar06}, and the recently proposed 3550-3650 MHz radar band for small cells usage, i.e. wireless BSs operating on low power \cite{FCC12_SmallCells}. 

In \cite{MultiCellMIMO_inter}, network MIMO, also known as multi-cell cooperation, has shown network-level interference management that significantly improve cellular systems performance. In network MIMO, multiple BSs cooperate their transmission to each user. Network MIMO can be reduced to MIMO broadcast channel (BC) in case of full cooperation between all BSs as shown in \cite{MIMO_capacity}. In another scenario, network MIMO can be reduced to a MIMO interference channel (MIMO-IFC) in case of absence of collaboration between BSs shown in \cite{MIMO_capacity,IA_DoF,MIMO_X_IA}. The general case is forming clusters of BSs that collaborate to transmit to a certain user   \cite{MIMO_X_IA,MIMO_downlink,MIMO_cluster, Ahmed_ITW10, Ahmed_INFOCOM10}.

\subsection{Our Contributions}

Our contributions in this paper are as follows:
\begin{itemize}
\item We propose a small singular values space projection method that facilitate coexistence between radar and cellular systems. 
\item We incorporate radar stations in the cellular system and show the equivalence of the new model to MIMO interference channel with general constraints (MIMO-IFC-GC) model for network MIMO with partial cooperation shown in \cite{linear_precoding_journal}.
\item We provide a suboptimal solution of the weighted sum-MSE minimization (WSMMSE) problem in Section \ref{sec:opt} for our proposed model. 
\end{itemize}

\textit{Notation}: Matrices and vectors are denoted by bold upper and lower case letters, respectively. Transpose and Hermitian operators are denoted by $(\cdot)^T$, and $(\cdot)^H$, respectively. 

The paper is organized as follows. Section~\ref{sec:sys_model} discusses system model for MIMO downlink system with radar and cellular coexistence. Moreover, it discussed the user message precoding at radar and cellular stations. Section \ref{sec:projection_matrix} describes how to construct small singular values space projection matrix. We show the equivalence of our proposed model with MIMO-IFC-GC in Section \ref{sec:equ_model}. Section \ref{sec:opt} contains the WSMMSE minimization problem under investigation and Section \ref{sec:min} presents its solution. Section~\ref{sec:conc} concludes the paper.

\section{System Model}\label{sec:sys_model}

We consider a MIMO downlink system with $L$ radar stations, including shipborne radars, forming a set $\mathcal{L}$, $M$ cellular base stations (BSs) forming a set $\mathcal{M}$, and $K$ mobile users forming a set $\mathcal{K}$, see Figure \ref{fig:system_model}. Each BS has $n_t$ antennas for transmission, each radar station has $n_{\text{\text{rad}}}$ antennas for transmission, and each mobile user has $n_r$ antennas. The $m$th BS has the messages for users set $\mathcal{K}_m \subseteq \mathcal{K}$ where $|\mathcal{K}_m| = K_m$. Similarly, the $l$th radar has the messages for users set $\mathcal{K}_l \subseteq \mathcal{K}$ where $|\mathcal{K}_l| = K_l$. Therefore, the $k$th user receives its intended message from a subset of $M_k$ BSs $\mathcal{M}_k \subseteq \mathcal{M}$ and a subset of radar stations $L_k$ radars $\mathcal{L}_k \subseteq \mathcal{L}$. In total, $k$th user receives its message from $L_k + M_k$ stations $\mathcal{M}_k \cup \mathcal{L}_k \subseteq \mathcal{M} \cup \mathcal{L}$. This channel is generally referred to as MIMO interference channel with partial message sharing, see \cite{linear_precoding_journal}. If $\mathcal{K}_m$, or $\mathcal{K}_l$, contains one user for each transmitter $m$, or $l$, then the model reduces to a standard MIMO interference channel (MIMO-IFC). When all transmitters cooperate in transmitting to all the users, i.e. $M_k = M$ and $L_k = L$, then we have MIMO broadcast channel (MIMO-BC), when number of some transmitters cooperate, i.e. $M_k < M$ or $L_k < L$, then we have multicast interference channel \cite{Ahmed_ITW10}. In this paper, we consider MIMO interference channel with partial message sharing (MIMO-IFC-PMS).

\begin{figure}
\centering
\includegraphics[width=\linewidth]{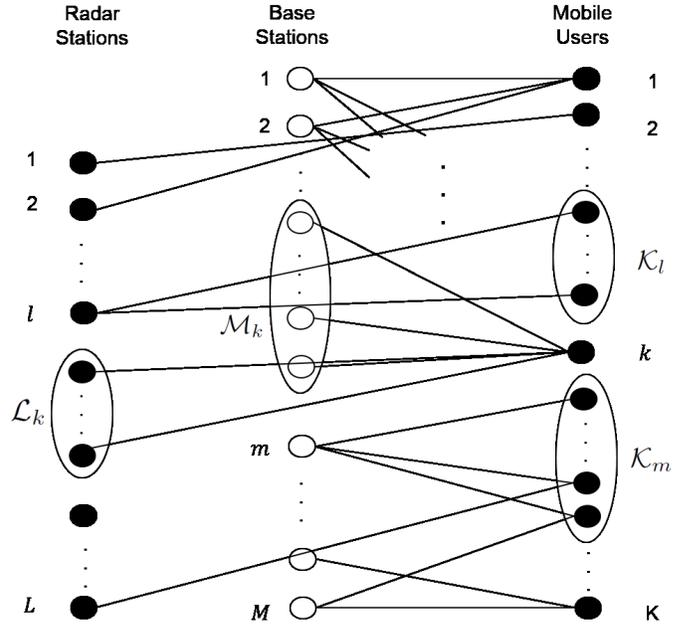}
\caption{System Model.} 
\label{fig:system_model}
\end{figure}

\subsection{Precoding}

 We define ${\mbf u_k = [u_{k,1} ...u_{k,d_k} ]^T} \in \mathbb{C}^{d_k}$  to represent the $d_k \leq \min({L_k} {n_{\text{\text{rad}}}} + {M_k} {n_t} ,n_r )$ independent streams sent to user $k$. It is assumed that $\mbf u_k \sim \mathcal{CN}(0,\mbf I)$. The data streams $\mbf u_k$ are known to all the cellular base stations in the set $\mathcal{M}_k$ and all the radar stations in the set $\mathcal{L}_k$. Assuming $l \in \mathcal{L}_k \subset \mathcal{L}$, the $l$th radar station precodes vector $\mbf u_k$ via a matrix $\mbf F_{k,l} \in \mathbb{C}^{ n_{\text{rad}} \times d_k} $, then projects it using projection matrix $\mbf P_{l} \in \mathbb{C}^{n_{\text{rad}} \times n_{\text{rad}}}$, which is described in Section \ref{sec:projection_matrix}. The signal $\tilde{\mbf x}_{l} \in \mathbb{C}^ {n_{\text{rad}}}$ sent by radar station and received by the user can be given as,
\begin{equation*}
\tilde{\mbf x}_{l}=\mbf P_{l}\sum_{k\in K_l} \mbf F_{k,l}\mbf u_k 
\end{equation*}
Assuming $P_l$ is the allowed radar power level to the communication system, then power constraint is given by
\begin{eqnarray*}
\mathbb{E}\left[\left\|{\tilde{\mbf x}_{l}}\right\|^2\right]&=& \text{tr} \left\{\mathbb{E}\left[\tilde{\mbf x}_{l} {\tilde{\mbf x}^{H}_{l}}\right]\right\}\\ \notag
&=&\sum_{k\in K_l}  \text{tr}\left\{\mbf P_{l} \mbf F_{k,l} \mbf F^{H}_{k,l} \mbf P^{H}_{l} \right\} \leq P_l, l=1,\cdots,L.
\end{eqnarray*}
Similarly, assuming $m \in \mathcal{M}_k $, for the $m$th base station, we have
\begin{equation}
\tilde{\mbf x}_m=\sum_{k\in K_m} \mbf F_{k,m} \mbf u_k.
\end{equation}
and,
\begin{eqnarray}
\mathbb{E}\left[\left\|{\tilde{\mbf x}_m}^2\right\|\right]&=& \text{tr} \left\{E\left[\tilde{\mbf x}_m {\tilde{\mbf x}^{H}_{m}}\right]\right\}\\ \notag
&=&\sum_{k\in K_m} \text{tr} \left\{\mbf F_{k,m} \mbf F^{H}_{k,m}\right\} \leq P_m, m=1,\cdots,M.
\end{eqnarray}
The $k$th user receives the following signal:
\begin{eqnarray}
{\mbf y}_k & = & \sum^{L}_{l=1} \tilde{\mbf H}_{k,l} {\tilde{\mbf x}_{l}} + \sum^{M}_{m=1} \tilde{\mbf H}_{k,m} {\tilde{\mbf x}_{m}} + {\tilde{\mbf n}_{k}}\\ \notag
           & = &\sum_{l\in \mathcal{L}_k} \tilde{\mbf H}_{k,l} \mbf P_{l} {{\mbf F}_{k,l}} {{\mbf u}_{k}}+\sum_{o\neq k}\sum_{j\in \mathcal{L}_o} \tilde{\mbf H}_{o,j} \mbf P_{o}{{\mbf F}_{o,j}} {{\mbf u}_{o}} \\ \notag
           & & + \sum_{m\in \mathcal{M}_k} \tilde{\mbf H}_{k,m} {{\mbf F}_{k,m}} {{\mbf u}_{k}} + \sum_{o\neq k}\sum_{j\in \mathcal{M}_o} \tilde{\mbf H}_{o,j} {{\mbf F}_{o,j}} {{\mbf u}_{o}} + {\tilde{\mbf n}_{k}}
\end{eqnarray}
where $\tilde{\mbf H}_{k,l} \in \mathbb{C}^{ n_r \times n_{\text{rad}}}$ and $\tilde{\mbf H}_{k,m} \in \mathbb{C}^{ n_r \times n_t}$ are the channel matrices between the $l$th radar station and $k$th user and the $m$th BS and $k$th user, respectively, and $\tilde{\mbf n}_{k}$ is noise where $\tilde{\mbf n}_{k} \sim \mathcal{CN}(0,\mbf I)$. 

\section{Small Singular Values Space Projection}\label{sec:projection_matrix}

In this section, we design our projection matrix such that steering the radar power in the direction of small singular values results in radar power that is in the order of communication transmit power, i.e. $P_l = \sigma_{\text{th}}P_{\text{rad}} = \Theta(P_{m}) \quad \forall m, l$ where $P_{\text{rad}}$ is the radar transmit power and $\sigma_{\text{th}}$ is a singular value threshold. As a result, the diminished radar power received at the communication system will not burnout or saturate the communication receivers. Given that the $l$th MIMO radar has channel state information of $\tilde{\mbf H}_{k,l} $ and the $k$th user. Let $(o)_l$ be the $o$th user in the set $\mathcal{K}_l$. Therefore the augmented channel matrix between the $l$th radar and the set of users $\mathcal{K}_l$ is given by
\begin{equation}
\tilde{\mbf H}_l = [\tilde{\mbf H}_{(1)_l,l}^{T} \tilde{\mbf H}_{(2)_l,l}^{T} \cdots \tilde{\mbf H}_{(K_l)_l,l}^{T}]^{T}.
\end{equation}
We proceed by first finding SVD of $\tilde{\mbf H}_l$, i.e., 
\begin{equation}
\tilde{\mbf H}_l = \mbf U_l \bsym \Sigma_l \mbf V_l^H.
\end{equation}
Now, let us define 
\begin{equation}
\bsymwt \Sigma_l \triangleq \text{diag} (\wt \sigma_{l,1}, \wt \sigma_{l,2}, \ldots, \wt \sigma_{l,p})
\end{equation}
where $p \triangleq \min (K_l n_r, n_{\text{rad}})$ and 
$\wt \sigma_{l,1} > \wt \sigma_{l,2} > \cdots > \wt \sigma_{l,q} > \wt \sigma_{l,q+1} = \wt \sigma_{l,q+2} = \cdots = \wt \sigma_{l,p} = 0$ are the singular values of $\tilde{\mbf H}_l $. Next, we define
\begin{equation}
\bsymwt {\Sigma}_l^\prime \triangleq \text{diag} (\wt \sigma_{l,1}^\prime,\wt \sigma_{l,2}^\prime, \ldots, \wt \sigma_{l,n_{\text{rad}}}^\prime)
\end{equation}
where
\begin{align}
\wt \sigma_{l,u}^\prime \triangleq
\begin{cases}
0, \quad \text{for} \; \wt \sigma_{l,u}^\prime > \sigma_{\text{th}},\\
1, \quad \text{for} \; \wt \sigma_{l,u}^\prime \leq \sigma_{\text{th}}.
\end{cases}
\end{align}
Using above definitions, we can now define our small singular values space projection matrix $\mbf P_l \triangleq \mbf V_l \bsymwt \Sigma_l^\prime \mbf V_l^H$. Note that $\mbf P_l \in \mbb C^{n_{\text{rad}} \times n_{\text{rad}}}$ is a projection matrix as it satisfies $\mbf P_l^H = (\mbf V_l \bsymwt \Sigma_l^\prime \mbf V_l^H)^H = \mbf P_l$ and $\mbf P_l^2 = \mbf V_l \bsymwt \Sigma_l \mbf V^H_l \times \mbf V_l \bsymwt \Sigma_l \mbf V^H_l = \mbf P_l$. 


\section{Equivalence with MIMO-IFC-GC}\label{sec:equ_model}

The MIMO-IFC-GC consists of $K$ transmitters with $m_{t,k}$ antennas at the the $k$th transmitter and $K$ receivers with $m_{r,k}$ antennas at the $k$th receiver. 
The $k$th receiver received signal is given by:
\begin{equation}
{\mbf y_k} = {\mbf H_{k,k}} {\mbf x_k} + \sum_{o\neq k}{\mbf H_{k,o}} {\mbf x_o} + \mbf n_k  
\end{equation}
where ${\mbf n}_{k}$ is additive complex Gaussian noise ${\mbf n}_{k} \sim \mathcal{CN}(0,\mbf I)$, $\mbf x_k \in \mathbb{C}^{m_{t,k}}$  are the inputs to receiver and $\mbf H_{k,o} \in \mathbb{C}^{m_{r,k}\times m_{t,k}}$ is the channel matrix between the $o$th transmitter and the $k$th receiver. The $k$th user intended information stream vector is $\mbf u_k \in \mathbb{C}^{d_k}$ where $d_k \leq \min(m_{t,k},m_{r,k})$ and ${\mbf u}_{k} \sim \mathcal{CN}(0,\mbf I)$. In this model, the $k$th user precoding matrix is given by $\mbf F_k \in \mathbb{C}^{m_{t,k} \times d_k}$ therefore $\mbf x_k = \mbf F_k \mbf u_k $. The input vectors $\mbf x_k$ have to satisfy both the $L$ generalized linear constraints given by
\begin{equation*}
\sum^{K}_{k=1} \text{tr}\left\{\mathbf{\Phi}_{k,l} \mathbb{E}\left[\textbf{x}_k {\textbf{x}^{H}_{k}}\right]\right\} =\sum^{K}_{k=1} \text{tr}\left\{\mathbf{\Phi}_{k,l} \textbf{F}_k {\textbf{F}^{H}_{k}}\right\} \leq P_l
\end{equation*}
and $M$ generalized linear constraints given by
\begin{equation}
\sum^{K}_{k=1} \text{tr}\left\{{\mathbf{\Phi}}_{k,m} E\left[\mbf x_k {\mbf x^{H}_{k}}\right]\right\} =\sum^{K}_{k=1} \text{tr}\left\{{\mathbf{\Phi}}_{k,m} \mbf F_k {\mbf F^{H}_{k}}\right\} \leq P_m
\end{equation}
for $\mathbf{\Phi}_{k,l} \in \mathbb{C}^{m_{t,k} \times m_{t,k}}$ and $l = 1,...,L$ and $\mathbf{\Phi}_{k,m} \in \mathbb{C}^{m_{t,k} \times m_{t,k}}$ and $m = 1,...,M$ are  weight matrices where $\sum^{L}_{l=1} \mathbf{\Phi}_{k,l} + \sum^{M}_{m=1} \mathbf{\Phi}_{k,m}$ are positive definite for all $k = 1,...,K$.  
\begin{lem}
Assume that the $o^{th}$ base station or radar station in subset $\mathcal{M}_k \cup \mathcal{L}_k$ is given by the index $(o)_k$ is informed about user $k$'s message. The MIMO-IFC-PMS is a special case of a MIMO-IFC-GC with $m_{r,k} = n_r$, $m_{t,k} = M_k n_t + L_k n_{\text{\text{rad}}}$, channel matrices
\begin{equation*}
\begin{aligned}
&\mbf H_{k,o} = & &[\tilde{\mbf H}_{k,(1)_o} \cdots \tilde{\mbf H}_{k,(M_k)_o} \tilde{\mbf H}_{k,(M_k+1)_o}\mbf P_{(M_k+ 1)_o} \\
&  & & \cdots \tilde{\mbf H}_{k,(M_k + L_k)_o}\mbf P_{(M_k + L_k)_o} ]
\end{aligned}
\end{equation*}
augmented precoding (beamforming) matrices,
\begin{equation*}
\mbf F_{k}={\left[{\mbf F}^{T}_{k,\left(1\right)_k} \cdots {\mbf F}^{T}_{k,\left(M_k\right)_k} {\mbf F}^{T}_{k,\left(1+M_k\right)_k} \cdots {\mbf F}^{T}_{k,\left(M_k+ L_k\right)_k}\right]}^T,
\end{equation*}
and  weight matrices $\mathbf{\Phi}_{k,l}$ with the $o^{th}$ $n_{\text{rad}} \times n_{\text{rad}}$ submatrix on the main diagonal is $\textbf{P}_l^H\textbf{P}_l$ , if $l = (o)_k$ and the rest of the matrix elements are zeros and $\mathbf{\Phi}_{k,m}$ with the $o^{th}$ $n_t \times n_t$ submatrix on the main diagonal is $\mbf I_{n_t}$ , if $m = (o)_k$ and the rest of the matrix elements are zeros. If $k \notin \mathcal{K}_m$ then $\mathbf{\Phi}_{k,m} = \mathbf{0}$ and if $l \notin \mathcal{K}_l$ then $\mathbf{\Phi}_{k,l} = \mathbf{0}$.
\end{lem}
\begin{proof}
It is easy to show by inspection.
\end{proof}
%
\begin{rem}
The effect of small singular values space projection on the MIMO-IFC-GC equivalent model is only in the augmented channel matrix or the weight matrices  $\mathbf{\Phi}_{k,l}$ which are inputs to the optimization problem as shown in section \ref{sec:opt}.
\end{rem}

\section{Optimization Problem}\label{sec:opt}

Given the equivalence between MIMO-IFC-GC and the proposed Network MIMO with partial cooperation for radar and cellular systems model. The rest of the paper shows how to use a modified version of the receiver in \cite{linear_precoding_journal} for our new proposed system model. The $k$th user  uses the equalization matrix $\mbf G_k \in \mathbb{C}^{d_k \times m_{r,k}}$ to estimate its message $\mbf u_k$ as
\begin{equation}
\hat{\mbf u}_k=\mbf G^{H}_{k} \mbf y_k.
\end{equation}
Therefore, the Mean Square Error (MSE)-matrix for user $k$ is given by
\begin{equation}
\mbf E_k=\mathbb{E}\left[\left(\hat{\mbf u}_k - \mbf u_k \right) \left(\hat{\mbf u}_k - \mbf u_k \right)^H \right].
\end{equation}
Using the equivalent MIMO-IFC-GC model, the MSE-matrix can be written as
\begin{eqnarray}\label{eqn:error}
\notag \mbf E_k=\mbf G^{H}_{k} \mbf H_{k,k} \mbf F_k \mbf F^{H}_{k} \mbf H^{H}_{k,k} \mbf G_{k} - \mbf G^{H}_{k} \mbf H_{k,k} \mbf F_{k} \\  
-\mbf F^{H}_{k} \mbf H^{H}_{k,k} \mbf G_{k} + \mbf G^{H}_{k} {\mathbf{\Omega}}_{k} \mbf G_{k}+\mbf I_{k}
\end{eqnarray}
where 
\begin{equation}
{\mathbf{\Omega}}_{k}=\mbf I + \sum_{o\neq k} \mbf H_{k,o} \mbf F_{o} \mbf F^{H}_{o} \mbf H^{H}_{k,o}.
\end{equation}
For each user $k$, the equalization matrices can be evaluated using the MMSE solution as
\begin{equation}\label{eqn:equalization}
\mbf G_k=\left(\mbf H_{k,k} \mbf F_k \mbf F_k^{H} \mbf H_{k,k}^{H} + \mathbf{\Omega}_k]\right)^{-1}\mbf H_{k,k} \mbf F_k.
\end{equation}
In this paper, we focus on the weighted sum-MSE minimization (WSMMSE) problem given by the following: 
\begin{equation}\label{eqn:optimizaiton}
\begin{aligned}
& \underset{\mbf F_{k},\mbf G_{k},k = 1,\cdots,K}{\text{min}} & & \sum^{K}_{k=1} \text{tr} \left\{\mbf W_k \mbf E_k \right\}\\
& \text{subject to} & &  \sum^{K}_{k=1} \text{tr}\left\{{\mathbf{\Phi}}_{k,m} \mbf F_k {\mbf F^{H}_{k}}\right\} \leq P_m, m=1,\cdots,M\\
& & &  \sum^{K}_{k=1} \text{tr}\left\{{\mathbf{\Phi}}_{k,l} \mbf F_k {\mbf F^{H}_{k}}\right\} \leq P_l, l=1,\cdots,L
\end{aligned}
\end{equation}
where $\mbf W_k \in \mathbb{C}^{d_k \times d_k}$ are the diagonal weight matrices with non-negative weights. 

\section{MMSE Minimization}\label{sec:min}

The extended MMSE interference alignment (eMMSE-IA) technique applied to an interference channel with per-transmitter power constraints and where each receiver is endowed with multiple antenna \cite{linear_precoding_journal}, is extended here to include radar system coexisting with communication system. The technique starts by an arbitrarily $\mbf F_k$. Then, at each iteration $j$ the equalization matrix $\mbf G^{\left(j\right)}_{k}$ is evaluated using (\ref{eqn:equalization})
resulting in
\begin{equation}
\mbf G^{\left(j\right)}_{k} = \left(\mbf H_{k,k} \mbf F^{\left(j-1\right)}_{k} \mbf F^{\left(j-1\right)H}_{k} \mbf H^{H}_{k,k} + \mathbf{\mathbf{\Omega}}^{\left(j-1\right)}_{k}\right)^{-1}\mbf H_{k,k} \mbf F^{\left(j-1\right)}_{k}
\end{equation}
where,
\begin{equation}
\mathbf{\Omega}^{\left(j-1\right)}_{k} = \mbf I + \sum_{o\neq k} \mbf H_{k,o} \mbf F^{\left(j-1\right)}_{o} \mbf F^{\left(j-1\right)H}_{o} \mbf H^{H}_{k,o}
\end{equation}
Given the matrices $\mbf G^{\left(j\right)}_{k}$, the optimization problem (\ref{eqn:optimizaiton}) becomes
\begin{equation}\label{eqn:optimizaiton2}
\begin{aligned}
& \underset{\mbf F_{k},k = 1,\cdots,K}{\text{min}} & &  \sum^{K}_{k=1} \text{tr} \{\mbf W_k \mbf E^{\left(j\right)}_{k} \}\\
& \text{subject to} & &  \sum^{K}_{k=1} \text{tr} \left\{{\mathbf{\Phi}}_{k,m} \mbf F_k {\mbf F^{H}_{k}}\right\} \leq P_m, \forall m\\
& & &  \sum^{K}_{k=1} \text{tr}\left\{{\mathbf{\Phi}}_{k,l} \mbf F_k {\mbf F^{H}_{k}}\right\} \leq P_l, \forall l
\end{aligned}
\end{equation}
where $\mbf E^{\left(j\right)}_{k}$ as in (\ref{eqn:error}) with $\mbf G^{\left(j\right)}_{k}$ instead of $\mbf G_{k}$. For a fixed $\mbf G^{\left(j\right)}_{k}$, the optimization problem in (\ref{eqn:optimizaiton2}) is convex and therefore there exists a unique global optimal solution for $\mbf F^{\left(j\right)}_{k}$. Using KKT conditions, we have 
\begin{equation}
\begin{aligned}
& \notag \mbf F^{(j)}_{k} = & &
\Bigl(\sum^{K}_{o=1} \mbf H^{H}_{o,k} \mbf G^{(j)}_{o} \mbf W_{o} \mbf G^{(j)H}_{o} \mbf H_{o,k} + \sum_{m} \mu_m \mathbf{\Phi}_{k,m} +  \\
& & & \sum_{l} \mu_l \mathbf{\Phi}_{k,l}\Bigr)^{-1} \times \mbf H^{H}_{k,k} \mbf G^{(j)}_{k} \mbf W_{k}
\end{aligned}
\end{equation}
where $\mu_m \geq 0$ are Lagrangian multipliers satisfying
\begin{equation}
\mu_m\left(\sum^{K}_{k=1} \text{tr} \left\{\mathbf{\Phi}_{k,m} \mbf F^{\left(j\right)}_{k} \mbf F^{\left(j\right)H}_{k}\right\}-P_m\right)=0
\end{equation} 
and $\mu_l \geq 0$ are Lagrangian multipliers satisfying
\begin{equation}
\mu_l\left(\sum^{K}_{k=1} \text{tr} \left\{\mathbf{\Phi}_{k,l} \mbf F^{\left(j\right)}_{k} \mbf F^{\left(j\right)H}_{k}\right\}-P_l\right)=0.
\end{equation} 
Using $\mbf F^{\left(j\right)}_{k}$, the iterative algorithm continues with the $(j + 1)$th iteration.

\section{Conclusion}\label{sec:conc}

In this paper, we considered a network MIMO with partial cooperation model where radar stations cooperate with cellular base stations (BS)s to deliver messages to intended mobile users. We designed a new projection matrix to mitigate radar stations interference to cellular system. In addition, this projection provides useful enhancement to the cellular system performance and QoS when radar stations cooperate in communication messages delivery. We showed that our constructed model, i.e. the radar stations act as BSs in the cellular system, is equivalent to a MIMO interference channel under generalized linear constraints (MIMO-IFC-GC). Finally, we provided a solution to minimize the weighted sum mean square error minimization problem (WSMMSE) with enforcing power constraints on both radar and cellular stations.  

\section{Future Work}

Comparison with other MSE minimization schemes will be considered in our future work. In addition, we plan to conduct a study on the sum rate maximization improvement for the new proposed model.

\bibliographystyle{ieeetr}
\bibliography{pub}

\end{document}